\documentclass[a4paper,UKenglish]{lipics-v2016}
 
\usepackage{microtype}
\usepackage{graphicx}

\usepackage{amssymb}
\usepackage{amsmath}
\usepackage{amsfonts}
\usepackage{nccmath} 
\usepackage{mathrsfs} 
\usepackage{eufrak} 
\usepackage[linesnumbered,ruled]{algorithm2e}
\usepackage{footnote}

\title{Metropolis-Hastings Algorithms for Estimating Betweenness Centrality in Large Networks}
\titlerunning{Estimating Betweenness Centrality} 

\author[1]{Mostafa Haghir Chehreghani}
\author[2]{Talel Abdessalem}
\author[3]{Albert Bifet}
\affil[1]{LTCI, T\'el\'ecom ParisTech, Universit\'e Paris-Saclay, Paris, France\\
  \texttt{mostafa.chehreghani@gmail.com}}
\affil[2]{LTCI, T\'el\'ecom ParisTech, Universit\'e Paris-Saclay, Paris, France\\
  \texttt{talel.abdessalem@telecom-paristech.fr}}
\affil[3]{LTCI, T\'el\'ecom ParisTech, Universit\'e Paris-Saclay, Paris, France\\
  \texttt{albert.bifet@telecom-paristech.fr}}  
  
\authorrunning{M. H. Chehreghani et. al.} 

\Copyright{Mostafa Haghir Chehreghani and Talel Abdessalem and Albert Bifet}



\EventEditors{---}
\EventNoEds{2}

\EventLongTitle{--}

\EventShortTitle{---}

\EventAcronym{---}

\EventYear{---}

\EventDate{---}

\EventLocation{Aalborg, Denmark}
\EventLocation{---}

\EventLogo{}

\SeriesVolume{---}

\ArticleNo{-}

\begin{document}

\maketitle

\begin{abstract}
{\em Betweenness centrality} is an important index
widely used in different domains such as social networks, traffic networks
and the world wide web.
However, even for mid-size networks that have only a few hundreds thousands vertices,
it is computationally expensive to compute exact betweenness scores.
Therefore in recent years, several approximate algorithms have been developed.
In this paper, first
given a network $G$ and a vertex $r \in V(G)$,
we propose a Metropolis-Hastings MCMC algorithm that samples from the
space $V(G)$ and estimates betweenness score of $r$.
The stationary distribution of our MCMC sampler is the optimal sampling
proposed for betweenness centrality estimation.
We show that our MCMC sampler provides an 
$(\epsilon,\delta)$-approximation, where the number of
required samples depends on the position of $r$ in $G$
and in many cases, it is a constant.
Then,
given a network $G$ and a set $R \subset V(G)$,
we present a Metropolis-Hastings MCMC sampler that samples from the
joint space $R$ and $V(G)$ and estimates relative betweenness scores of the vertices in $R$.
We show that for any pair $r_i, r_j \in R$,
the ratio of the expected values of the estimated relative betweenness scores
of $r_i$ and $r_j$ respect to each other is equal to
the ratio of their betweenness scores.
We also show that our joint-space MCMC sampler provides an 
$(\epsilon,\delta)$-approximation
of the relative betweenness score of $r_i$ respect to $r_j$,
where the number of required samples depends on the position of $r_j$ in $G$
and in many cases, it is a constant.
\end{abstract}

\section{Introduction}
\label{sec:introduction}

\textit{Centrality} is a structural property of vertices (or edges) in the network
that quantifies their relative importance.
For example, it determines the importance of a person within a social network,
or a road within a road network.
Freeman \cite{jrnl:Freeman} introduced and defined 
\textit{betweenness centrality} of a vertex
as the number of shortest paths
from all (source) vertices to all others that pass through that vertex.
He used it for measuring the control of a human over 
the communications among others in a social network \cite{jrnl:Freeman}.
Betweenness centrality is also used in some well-known algorithms for 
clustering and community detection in social and information networks \cite{jrnl:Girvan}.

Although there exist polynomial time and space algorithms for betweenness centrality estimation,
the algorithms are expensive in practice.
Currently, the most efficient existing exact method is Brandes's algorithm \cite{jrnl:Brandes}.
Time complexity of this algorithm is $O(nm)$ for unweighted graphs and
$O(nm+n^2 \log n)$ for weighted graphs with positive weights,
 where $n$ and $m$ are the number of vertices and the number of edges of the network, respectively.
This means exact betweenness centrality computation is not
applicable, even for mid-size networks.
However, there exist observations that may improve computation of betweenness scores
in practice.
\begin{itemize} 
\item 
First, in several applications it is sufficient to
compute betweenness score of only one or a few vertices.
For instance, this
index might be computed for only core vertices of communities in social/information
networks \cite{jrnl:Wang}
or for only hubs in communication networks.
Note that these vertices are not necessarily
those that have the highest betweenness scores.
Hence, algorithms that identify vertices
with the highest betweenness scores \cite{Riondato2016} are not applicable.
While exact computation of this index for one vertex is not easier than 
that for all vertices,
Chehreghani \cite{DBLP:journals/cj/Chehreghani14} and later Riondato and Kornaropoulos \cite{Riondato2016}
respectively showed that this index can be estimated more effectively for one
arbitrary vertex and for $k$ vertices that have the highest scores.
\item
Second, 
in practice,
instead of computing betweenness scores,
it is usually sufficient to
\textit{compute betweenness ratios} or
\textit{rank vertices} according to their betweenness scores \cite{Riondato2016}.
For example, Daly and Haahr \cite{DBLP:journals/tmc/DalyH09} exploited betweenness ratios for
finding routes that provide good
delivery performance and low delay in Mobile Ad hoc Networks.
In the Girvan and Newman's iterative algorithm 
for community detection \cite{jrnl:Girvan},
at each iteration it is required to detect 
the edge with the highest betweenness score 
and remove it from the network.
The other application discussed in \cite{DBLP:journals/corr/AgarwalSCI14}, 
is handling cascading failures.
\end{itemize}

While the above mentioned observations do not yield a better algorithm
when exact betweenness scores are used,
they may improve approximate algorithms.
In particular, algorithms such as
\cite{DBLP:journals/cj/Chehreghani14} and \cite{RiondatoKDD20116}
have used the first observation to propose algorithms that can estimate
betweenness of a single vertex or a few vertices better than estimating it for all vertices.

In the current paper, we exploit both observations to design more effective
approximate algorithms.
In the first problem studied in this paper,
we assume that we are given a vertex $r \in V(G)$
and we want to estimate its betweenness score.
In the second problem,
we assume that we are given a set $R \subset V(G)$
and we want to estimate
the ratio of
betweenness scores of vertices in $R$.
The second problem is formally defined as follows:
given a graph $G$ and a set $R \subset V(G)$,
for any two vertices $r_i$ and $r_j$ in $R$, 
we want to estimate the \textit{relative betweenness score} of $r_i$ respect to $r_j$,
denoted by $ BC_{r_j}(r_i)$
(see Equation~\ref{eq:relativebetweenness} of Section~\ref{sec:mcmcsampler}
for the formal definition of relative betweenness score).
The ratio of our estimations of $BC_{r_j}(r_i)$ and $BC_{r_i}(r_j)$
is equal to the ratio of betweenness scores of $r_i$ and $r_j$.
To address these problems, we use a technique which is new in the context of network analysis,
but has already been used in statistical physics to estimate free energy differences \cite{bennett76}.
Our technique is based on a sampling method that belongs to a 
specific class of Monte Carlo Markov Chain (MCMC) algorithms, called Metropolis Hastings \cite{mengersen1996}.

In this paper, our key contributions are as follows.
\begin{itemize}
\item
Given a graph $G$ and a vertex $r \in V(G)$,
in order to estimate betweenness score of $r$,
we develop a MCMC sampler that samples from the space $V(G)$.
We show that our MCMC sampler can provide an 
$(\epsilon,\delta)$-approximation
of the betweenness score of $r$,
where $\epsilon \in \mathbb R^+$ and $\delta \in (0,1)$.
In particular, we prove that 
$\mathbb P\left[ |BC(r)-\ddddot{BC}(r)| > \epsilon \right] \leq \delta ,$
where $BC(r)$ is the betweenness score
of $r$ and 
$\ddddot{BC}(r)$ is its estimated value.
Unlike existing work, our samples are non-iid and
the stationary distribution of our MCMC sampler
is the \textit{optimal probability distribution} proposed in \cite{DBLP:journals/cj/Chehreghani14}.
Moreover, the number of samples required by our sampler
depends on the position of the vertex in the graph and
in several cases, it is a constant.
\item 
Given a graph $G$ and a set $R \subset V(G)$,
in order to estimate relative betweenness scores
of all pairs of vertices in $R$,
we develop a MCMC sampler that samples from the joint space $R$ and $V(G)$.
This means each sample (state) in our MCMC sampler is a pair
$\left\langle \mathfrak r,\mathfrak v \right\rangle$,
where $\mathfrak r \in R$
and $\mathfrak v \in V(G)$.
For any two vertices $r_i,r_j \in R$, 
we show that our joint-space MCMC sampler provides an 
$(\epsilon,\delta)$-approximation
of the relative betweenness score of $r_i$ respect to $r_j$,
where the number of required samples depends only on the position of
$r_j$ in the graph and in several cases, it is a constant.
\end{itemize}

The rest of this paper is organized as follows. 
In Section~\ref{sec:preliminaries},
preliminaries and necessary definitions related to betweenness centrality
and MCMC algorithms are introduced.
A brief overview on related work is given in Section~\ref{sec:relatedwork}.
In Section~\ref{sec:algorithm}, we
present our MCMC samplers and their analysis.
Finally, the paper is concluded in Section~\ref{sec:conclusion}.

\section{Preliminaries}
\label{sec:preliminaries}

In this section, we present definitions and notations widely used in the paper.
We assume that the reader is familiar with basic concepts in graph theory.
Throughout the paper, $G$ refers to a graph (network).
For simplicity, we assume that $G$ is a undirected, connected and loop-free graph without multi-edges.
Throughout the paper, we assume that $G$ is an unweighted graph,
unless it is explicitly mentioned that $G$ is weighted.
$V(G)$ and $E(G)$ refer to the set of vertices and the set of edges of $G$, respectively.
For a vertex $v \in V(G)$, by $G \setminus v$ we refer to the set of connected graphs
generated by removing $v$ from $G$.
In the following,
first in Section~\ref{section:betweenness} we introduce definitions and concepts related to betweenness centrality.
Then in Section~\ref{section:mcmc}, we present a brief overview of MCMC and Metropolis-Hastings algorithms. 


\subsection{Betweenness centrality}
\label{section:betweenness}
A \textit{shortest path} (also called a \textit{geodesic path})
between two vertices $u,v \in V(G)$ is a path
whose size is minimum, among all paths between $u$ and $v$.
For two vertices $u,v \in V(G)$,
we use $d(u,v)$, to denote the size (the number of edges) of a shortest path connecting $u$ and $v$.
By definition, $d(u,v)=0$ and $d(u,v) = d(v,u)$.
For $s,t \in V(G)$, $\sigma_{st}$ denotes the number of shortest paths between $s$ and $t$, and
$\sigma_{st}(v)$ denotes the number of shortest paths between $s$ and $t$ that also pass through $v$. 
{\em Betweenness centrality} of a vertex $v$ is defined as:
\begin{equation}
BC(v)= \frac{1}{|V(G)| \cdot \left( |V(G)|-1 \right)} \sum_{s,t \in V(G) \setminus \{v\}} \frac{\sigma_{st}(v)}{\sigma_{st}}.
\end{equation}

A notion which is widely used for counting the number of shortest paths in a graph is the directed acyclic graph (DAG)
containing all shortest paths starting from a vertex $s$ (see e.g., \cite{jrnl:Brandes}).
In this paper, we refer to it as the \textit{shortest-path-DAG}, or \textit{SPD} in short, rooted at $s$.
For every vertex $s$ in graph $G$,
the \textit{SPD} rooted at $s$ is unique,
and it can be computed in $O(|E(G)|)$ time for unweighted graphs
and in $O\left(|E(G)|+|V(G)|\text{ log }|V(G)|\right)$ time
for weighted graphs with positive weights \cite{jrnl:Brandes}.

Brandes \cite{jrnl:Brandes} introduced the notion of the \textit{dependency score}
of a vertex $s \in V(G)$ on a vertex $v \in V(G) \setminus \{s\}$, which is defined as:
\begin{equation}
\delta_{s\bullet}(v)=\sum_{t \in V(G) \setminus \{v,s\}} \delta_{st}(v) 
\end{equation}
where
$\delta_{st}(v) = \frac {\sigma_{st}(v)}{\sigma_{st}}.$
We have:
\begin{equation}
BC(v)= \frac{1}{|V(G)| \cdot \left( |V(G)|-1 \right)} \sum_{s \in V(G) \setminus \{v\}} \delta_{s\bullet}(v).
\end{equation}

Brandes \cite{jrnl:Brandes} showed that dependency scores of a
source vertex on different vertices in the network can be computed using a recursive relation,
defined as the following:
\begin{equation}
\delta_{s\bullet}(v)=\sum_{w:v \in P_s(w)} \frac{\sigma_{sv}}{\sigma_{sw}}(1+\delta_{s\bullet}(w)),
\label{eq:recursive}
\end{equation}
where $P_s(w)$ is defined as
$\{ u \in V(G): \{u,w\} \in E(G) \wedge d_{G}(s,v)=d_{G}(s,u)+1 \}.$
In other words, $P_s(w)$ contains the set of all parents (predecessors) of $w$ in the SPD rooted at $s$.
As mentioned in \cite{jrnl:Brandes}, given the SPD rooted at $s$,
for unweighted graphs and weighted graphs with positive weights,
dependency scores of $s$ on all other vertices can be computed in
$O(|E(G)|)$ time and $O(|V(G)|\log|V(G)| + |E(G)|)$ time, respectively.   

\subsection{MCMC and Metropolis-Hastings algorithms}
\label{section:mcmc}
In this section, we briefly review some basic notions and definitions in MCMC and Metropolis-Hasting algorithms.
For a comprehensive study, the interested reader can refer to e.g., \cite{brooks2011handbook}.
A Markov chain is a sequence of dependent random variables (states)
such that the probability distribution of each variable given the past variables depends
only on the last variable.
A Markov chain has \textit{stationary distribution}
if the conditional distribution of the $k+1^{\text{th}}$ state
given the $k^{\text{th}}$ state does not depend on $k$.

Let $\mathbb P[x]$ be a probability distribution defined on the random variable $x$.
When the function $f(x)$, which is proportional to the density of $\mathbb P[x]$,
can be efficiently computed,
the Metropolis-Hastings algorithm is used to draw samples
from $\mathbb P[x]$.
In a simple form (with symmetric proposal distribution),
the Metropolis-Hastings algorithm
first chooses an arbitrary initial state $x_0$.
Then, iteratively: 
\begin{itemize}
\item Let $x$ be the current state.
It generates a candidate $x'$ using the \textit{proposal distribution} $q(x'|x)$.
\item It moves from $x$ to $x'$ with probability $\min\left\{ 1,\frac{f(x')}{f(x)} \right\}$.
\end{itemize}
The \textit{proposal distribution} $q(x'|x)$ 
defines the conditional probability of proposing a state $x'$ given the state $x$.
In the \textit{Independence Metropolis-Hastings algorithm},
$q(x'|x)$ is independent of $x$, i.e., $q(x'|x) = g(x')$.
A sufficient but not necessary condition 
for the existence of stationary distribution in a Metropolis-Hastings algorithm
is the \textit{detailed balance} condition.
It says that for each pair of states $x$ and $x'$, the probability of being in state $x$
and moving to state $x'$ is equal to the probability of being in state $x'$ and moving to state $x$.
By $TR(x'|x)$, we denote the probability of moving from state $x$ to state $x'$.


 \section{Related work}
\label{sec:relatedwork}



Brandes \cite{jrnl:Brandes} introduced an efficient algorithm 
for computing betweenness centrality of a vertex,
which is performed in 
$O( |V(G)| |E(G)| )$ and $O(|V(G)| |E(G)| + |V(G)|^2 \log |V(G)|)$ 
times for unweighted and weighted networks with positive weights, respectively.
{\c{C}}ataly{\"{u}}rek et. al. \cite{DBLP:conf/sdm/CatalyurekKSS13} presented the
{\em compression} and {\em shattering} techniques to improve efficiency of Brandes's algorithm
for large graphs. During {\em compression}, vertices with known betweenness scores are removed from the graph and
during {\em shattering}, the graph is partitioned into smaller components.
Holme \cite{jrnl:Holme} showed that betweenness centrality of a vertex is
highly correlated with the fraction of time that the vertex is occupied
by the traffic of the network.
Barthelemy \cite{jrnl:Barthelemy} showed that many scale-free networks \cite{jrnl:Barabasi}
have a power-law distribution of betweenness centrality.


\subsection{Generalization to sets}
Everett and Borgatti \cite{jrnl:Everett} defined \textit{group betweenness centrality}
as a natural extension of betweenness centrality for sets of vertices.
Group betweenness centrality of a set is defined as the
number of shortest paths passing through at least one of the vertices in the set \cite{jrnl:Everett}.
The other natural extension of betweenness centrality is \textit{co-betweenness centrality}.
Co-betweenness centrality is defined as the number of shortest paths passing through all vertices in the set.
Kolaczyk et. al. \cite{jrnl:Kolaczyk} presented an $O(|V(G)|^3)$ time algorithm for
co-betweenness centrality computation of sets of size 2.
Chehreghani \cite{conf:cbcwsdm} presented efficient algorithms for co-betweenness centrality computation
of any set or sequence of vertices in weighted and unweighted networks.
Puzis et. al. \cite{jrnl:PuzisPhysRev} proposed an $O(|K|^3)$ time algorithm for 
computing successive group betweenness centrality,
where $|K|$ is the size of the set.
The same authors in \cite{jrnl:PuzisAIComm} presented two algorithms
for finding \textit{most prominent group}.
A \textit{most prominent group} of a network is a set vertices of minimum size,
so that every shortest path in the network passes through at least one of the vertices in the set.
The first algorithm is based on a heuristic search and
the second one is based on iterative greedy choice of vertices.

\subsection{Approximate algorithms}

%
%

Brandes and Pich
\cite{jrnl:Brandes3} proposed an approximate algorithm based on
selecting $k$ source vertices and
computing dependency scores of them on the other vertices in the graph.
They used various strategies for selecting the source vertices, including: 
MaxMin, MaxSum and MinSum \cite{jrnl:Brandes3}.
In the method of \cite{proc:Bader}, some source vertices are selected uniformly at random,
and their dependency scores are computed and scaled for all vertices.
Geisberger et. al. \cite{conf:Geisberger} presented an algorithm
for approximate ranking of vertices based on their betweenness scores.
In this algorithm, the method for aggregating dependency
scores changes so that vertices do not profit from being near the selected source vertices.
Chehreghani \cite{DBLP:journals/cj/Chehreghani14} 
proposed a randomized framework for
unbiased estimation of the betweenness score of a single vertex.
Then, 
to estimate betweenness score of vertex $v$,
he proposed a non-uniform sampler, defined as follows:
$\mathbb P[s] = \frac{d(v,s)}{\sum_{u \in V(G)\setminus \{v\}} d(v,u)}$,
where $s \in V(G) \setminus \{v\}$.

Riondato and Kornaropoulos \cite{Riondato2016}
presented shortest path samplers for estimating betweenness centrality of
all vertices or the $k$ vertices 
that have the highest betweenness scores in a graph.
They determined the number of samples
needed to approximate the betweenness with the desired accuracy
and confidence by means of the VC-dimension theory \cite{vc-ucrfep-71}.
Recently, Riondato and Upfal \cite{Riondato2016} introduced algorithms 
for estimating betweenness scores of all vertices in a graph.
They also discussed a variant of the algorithm that finds the top-$k$ vertices.
They used Rademacher average \cite{Shalev-Shwartz:2014:UML:2621980} to determine the number
of required samples.
Finally, Borassi and Natale \cite{DBLP:conf/esa/BorassiN16} presented the KADABRA algorithm,
which uses balanced bidirectional BFS (bb-BFS) to sample shortest paths.
In bb-BFS, a BFS is performed from each of the two endpoints $s$ and $t$,
in such a way that they are likely
to explore about the same number of edges.

\subsection{Dynamic graphs}
Lee et. al. \cite{proc:www}
proposed an algorithm to efficiently update betweenness centralities of
vertices when the graph obtains a new edge.
They reduced the search space by finding a candidate set of 
vertices whose betweenness centralities
can be updated.
Bergamini et. al. \cite{DBLP:conf/alenex/BergaminiMS15} presented approximate algorithms
that update betweenness scores 
of all vertices when an edge is inserted or 
the weight of an edge decreases.
They used the algorithm of \cite{Riondato2016} as the building block.
Hayashi et. al. \cite{DBLP:journals/pvldb/HayashiAY15}
proposed a fully dynamic algorithm for estimating betweenness centrality of all vertices in a large
dynamic network.
Their algorithm is based on two data structures: {\em hypergraph sketch} that
keeps track of SPDs,
and {\em two-ball index} that helps
to identify the parts of hypergraph sketches that require updates.

\section{MCMC Algorithms for estimating betweenness centrality} 
\label{sec:algorithm}

In this section, we present our MCMC sampler for
estimating betweenness score of a single vertex; and
our joint-space MCMC sampler for
estimating relative betweenness scores of vertices in a given set.


\subsection{Betweenness centrality as a probability distribution}

Chehreghani \cite{DBLP:journals/cj/Chehreghani14}
presented a randomized algorithm that admits a probability mass function
as an input parameter.
Then, he proposed an optimal sampling technique
that computes betweenness score of a vertex $r \in V(G)$ with error $0$.
In optimal sampling, each vertex $v$ is chosen with probability
\begin{equation}
\label{eq:optimalsampling}
\mathbb P_r\left[v\right]=\frac{\delta_{v\bullet(r)}}{\sum_{v' \in V(G)} \delta_{v'\bullet(r)}} 
\end{equation}
In other words, for estimating betweenness score of vertex $r$,
each source vertex $v \in V(G) $ whose dependency score on $r$ is greater than $0$,
is chosen with probability $\mathbb P\left[v\right]$
defined in Equation~\ref{eq:optimalsampling}.

In the current paper, for $r \in V(G)$ we want to estimate
$BC(r)$ and also for all pairs of vertices
$r_i,r_j$ in a set $R  \subset V(G)$,
the ratios $\frac{BC(r_i)}{BC(r_j)}$.
For this purpose, we follow a {\em source vertex sampling} procedure where for each vertex $r$, we consider 
$\mathbb P_{r} \left[\cdot \right]$ defined in Equation~\ref{eq:optimalsampling}
as the target probability distribution used to sample vertices $v\in V(G)$.
It is, however, computationally expensive to calculate
the normalization constant $\sum_{v' \in V(G)} \delta_{v'\bullet(r)}$ in
Equation~\ref{eq:optimalsampling},
as it gives the betweenness score of $r$.
However, for two vertices $v_1, v_2 \in V(G)$,
it might be feasible to compute the ratio
$\frac{\mathbb P_{r} \left[v_1\right]}{\mathbb P_{r} \left[v_2\right]}=
\frac{\delta_{v_1\bullet}(r)}{\delta_{v_2\bullet}(r)}$,
as it can be done in $O(|E(G)|)$ time for unweighted graphs
and in $O(|E(G)|+|V(G)|\log |V(G)|)$ time for weighted graphs with positive weights. 
This motivates us to propose Metropolis-Hastings
sampling algorithms that for a vertex $r$,
sample each vertex $v \in V(G)$
with the probability distribution $\mathbb P_{r}\left[v\right]$
defined in Equation~\ref{eq:optimalsampling}.


\subsection{A single-space MCMC sampler}
\label{sec:mcmcsampler0}

In this section, we propose 
a MCMC sampler, defined on the space $V(G)$,
to estimate betweenness centrality of a single vertex $r$.
Our MCMC sampler consists of the following steps:
\begin{itemize}
\item 
First, we choose a vertex $\mathfrak v_0 \in V(G)$,
as the initial state,
uniformly at random. 
\item
Then, at each iteration $t$, $1 \leq t \leq T$:
\begin{itemize}
\item
Let $\mathfrak v(t)$ be the current state of the chain.
\item
We choose vertex $\mathfrak v'(t) \in V(G)$,
uniformly at random.
\item
With probability
\begin{equation}
\min \left\{ 1,\frac{\delta_{\mathfrak v'(t)\bullet}(r)}{\delta_{\mathfrak v(t) \bullet}(r)} \right\} 
\end{equation}
we move from state $\mathfrak v(t)$ to the state
$\mathfrak v'(t)$.
\end{itemize}
\end{itemize}

The sampler is an iterative procedure where at each iteration $t$,
one transition may occur in the Markov chain.
Let $M$ be the multi-set
(i.e., the set where repeated members
are allowed)
of samples (states) accepted by our sampler.
In the end of sampling, betweenness score of $r$ is estimated as
\begin{equation}
\ddddot{BC}(r)=\frac{1}{(T+1)(|V(G)|-1)}\sum_{v \in M}\sum_{u \in V(G)\setminus\{v\}}\frac{\sigma_{vu}(r)}{\sigma_{vu}}.
\end{equation}

This estimation does not give an unbiased estimation of $BC(r)$,
however as we discuss below,
by increasing $T$,
$\ddddot{BC}(r)$ can be arbitrarily close to $BC(r)$. 
In the rest of this section, we show that our MCMC sampler provides
an $(\epsilon,\delta)$-approximation of $BC(r)$,
where $\epsilon \in \mathbb R^+$ and $\delta \in (0,1)$.
Our proof is based on a theorem presented
in \cite{Latuszynski_Miasojedow_Niemiro_2012} for the analysis of MCMC samples
and a theorem presented in \cite{mh} about the uniformly ergodicity of
Independence Metropolis-Hastings algorithms.
We first present these two theorems.

%

\begin{enumerate}
\item 
Suppose that we have a uniformly Ergodic MCMC with the proposal density $q$,
which is
defined on the space $S$ and satisfies the following condition:
there exist a constant $\lambda > 0$ and a probability measure $\varphi$
such that: 
\begin{equation}
\label{eq:uniformlyergodic}
q(\cdot \mid x) \geq \lambda \times \varphi(\cdot)  \text{,    for every } x \in S.
\end{equation}
Let
$f: S \rightarrow \mathbb R$ be some (Borel measurable) function defined from $S$ to $\mathbb R$,
and $\theta$ be $\frac{1}{|S|}\sum_{x \in S} f(x)$.
We define: \[\parallel f \parallel_{sp} = \sup_{x\in S} f(x) - \inf_{x \in S} f(x).\]
Now, suppose that we draw samples $\left\{x_1,x_2,\ldots,x_n\right\}$ from the MCMC.
Let $\hat{\theta}_n$ be $\frac{1}{n} \sum_{i=1}^n f(x_i)$. 
The authors of \cite{Latuszynski_Miasojedow_Niemiro_2012}
presented the following extension of Hoeffding's inequality
for these MCMC samples:
\begin{align}
\label{eq:expinequalty}
\mathbb P\left[ |\hat{\theta}_n-\theta| > \epsilon \right]  
\leq 2 \exp{ \left\{ - \frac{n-1}{2} 
\left( \frac{2\lambda}{\parallel f \parallel_{sp}}\epsilon -
\frac{3}{n-1} \right)^2  \right\} },
\end{align}
where $\epsilon \in \mathbb R^+ $.


\item
The authors of \cite{mh} showed that
an Independence Metropolis-Hastings algorithm with the
stationary probability distribution $P\left[ \cdot \right]$ and
proposal distribution $g(\cdot)$
is uniformly ergodic if
there exists a constant $\beta$ such that
\begin{equation}
\label{eq:independenceuniformlyergodic}
g(\cdot) \geq \beta \times \mathbb P\left[ \cdot \right].
\end{equation}
In this case,
$\lambda$ and $\varphi(\cdot)$ in Inequation~\ref{eq:uniformlyergodic}
will be $\beta$ and $\mathbb P\left[ \cdot \right]$, respectively
(and of course $q(\cdot \mid x)$ will be $g(\cdot)$).
\end{enumerate}

Now, for a vertex $r \in V(G)$, 
we use Inequalities~\ref{eq:uniformlyergodic}, \ref{eq:expinequalty}
and \ref{eq:independenceuniformlyergodic}
to derive an error bound for $\ddddot{BC}(r)$.

\begin{theorem}
\label{theorem:ourerrorbound}
Let
$\overline{\delta(r)}$ be
the average of dependency scores of vertices 
in $V(G)$ on $r$, i.e.,
\[ \overline{\delta(r)}= \frac{\sum_{v\in V(G)}\delta_{v\bullet}(r)}{|(V(G)|}.\]
Suppose that 
there exists some value $\mu(r)$ such that 
for each vertex $v \in V(G)$, the following holds:
\begin{equation}
\label{eq:averagebcbound}
\delta_{v\bullet}(r) \leq \mu(r) \times \overline{\delta(r)}.
\end{equation}
Then, 
for a given $\epsilon \in \mathbb R^+$,
by our MCMC sampler
and starting from any arbitrary initial state, we have
\begin{align}
\label{eq:errorbound}
\mathbb P\left[ |\ddddot{BC}(r)-BC(r)| > \epsilon \right] 
\leq 2 \exp{ \left\{ - \frac{T}{2} 
\left( \frac{2\epsilon}{\mu(r)} -
\frac{3}{T} \right)^2  \right\} }.
\end{align}
\end{theorem}

\begin{proof}
Our MCMC sampler is an 
Independence Metropolis Hastings MCMC whose proposal density is $\frac{1}{|V(G)|}$.
Inequation \ref{eq:averagebcbound} yields that
\[\frac{\delta_{v\bullet}(r)}{\overline{\delta(r)} } =
\frac{|V(G)| \cdot \delta_{v\bullet}(r)}{\sum_{v\in V(G)}\delta_{v\bullet}(r)}
\leq \mu(r),\]
which yields 
\begin{equation}
\frac{1}{|V(G)|}  \geq \frac{1}{\mu(r)} \times
\frac{\delta_{v\bullet}(r)}{\sum_{v\in V(G)}\delta_{v\bullet}(r)}.
\end{equation}

Hence, our MCMC sampler
satisfies Inequation~\ref{eq:independenceuniformlyergodic},
where the proposal density is
$\frac{1}{|V(G)|}$,
$\mathbb P \left[ \cdot \right]$ is $\mathbb P_{r}\left[v\right]$
defined in Equation~\ref{eq:optimalsampling}
and $\beta$ is $\frac{1}{\mu(r)}$.
Now, if we apply Inequation~\ref{eq:expinequalty}
with:
$n = T+1$,
$S =V(G)$,
$x$ a vertex $v$ in $V(G)$ and
$f(v) =
\frac{1}{|V(G)|-1} \sum_{u\in V(G)\setminus \{v\}} \frac{\sigma_{vu}(r)}{\sigma_{vu}}$,
we get:
$\parallel f \parallel_{sp} = 1$,
$\theta = \frac{1}{|V(G)|} \sum_{v \in V(G)} f(v) = BC(r)$,
$\hat{\theta}_n = \frac{1}{T+1} \sum_{v \in M} f(v) = 
\ddddot{BC}(r)$,
and
$\lambda = \frac{1}{\mu(r)}$.
Putting these values into Inequation~\ref{eq:expinequalty},
we get Inequation~\ref{eq:errorbound}. 
\end{proof}

Note that since Inequation~\ref{eq:expinequalty} holds
for any arbitrary initial distribution, 
Inequation~\ref{eq:errorbound} does not depend on the initial state.
Furthermore, since in Inequation~\ref{eq:expinequalty}
it is not required to discard an initial part of the chain, called {\em burn-in},
Inequation~\ref{eq:errorbound} holds without need for burn-in. 
$T$ is usually large enough so that we can approximate $\frac{3}{T}$ by $0$. 
Hence, Inequation~\ref{eq:errorbound} yields that
for given values $\epsilon \in \mathbb{R}^+$ and $\delta \in (0,1)$,
if $T$ is chosen such that
\begin{equation}
\label{eq:epsilondeltaT}
T  \geq  \frac{{\mu(r)}^2}{2\epsilon^2}   \ln{\frac{2}{\delta}} 
\end{equation}
then, our MCMC sampler can estimate the betweenness score of $r$ 
within an additive error $\epsilon$ with probability $\delta$.

Similar to the
error bounds presented in e.g., \cite{Riondato2016}, \cite{RiondatoKDD20116} and \cite{DBLP:conf/esa/BorassiN16},
worst case time complexity of processing each sample in our algorithm is $O(|E(G)|)$
for unweighted graphs
(and $O(|V(G)| \log |V(G)| + |E(G)|)$ for weighted graphs with positive weights).
The number of samples required by our algorithm to estimate betweenness score of vertex $r$
depends on the parameter $\mu(r)$.
The value of $\mu(r)$ depends on the {\em position} of $r$ in the graph.
In the rest of this section, we show that in several cases
(where $r$ is an {\em important} vertex in the graph),
$\mu(r)$ is a constant and as a result,
our MCMC sampler can give an $(\epsilon,\delta)$-approximation of betweenness score of $r$
using only a constant number of samples.

\begin{theorem}
\label{theorem:constantmu}
Let $G$ be a graph, $r \in V(G)$ 
and $C=\{C_1,\ldots,C_l\}$ be the set $G \setminus r$.
For an $i \in [1 .. l]$,
let $V_i$ refer to the sum of the number of vertices of the graphs in the set $C \setminus C_i$.
If for all $i \in [1 .. l]$ we have: $V_i =\Theta(|V(G)|)$,
then $\mu(r)$ is a constant. 
\end{theorem}
\begin{proof}
Let
$m$ be a vertex in $V(G) \setminus \{r\}$ such that for each vertex $v \in V(G) \setminus \{r\}$,
the following holds: $\delta_{m\bullet}(r) \geq \delta_{v\bullet}(r)$.
In order to prove that $\mu(r)$ is a constant, we need to prove that
$\frac{\delta_{m\bullet}(r)}{\overline{\delta(r)}}$ is a constant.
Without loss of generality, assume that $m$ belongs to $C_1$.
On the one hand, 
all shortest paths between $m$ and all vertices in $C_2,\ldots,C_l$ pass through $r$
and none of the shortest paths between $m$ and any vertex in $C_1$ pass through $r$.
Hence, dependency score of $m$ on $r$ is $|V(C_2)|+\ldots+|V(C_l)|$.
On the other hand,
all shortest paths between $v_1 \in V(C_i)$ and $v_2 \in V(C_j)$ so that $i \neq j$
pass through $r$
and none of the shortest paths 
between $v_3 \in V(C_i)$ and $v_4 \in V(C_j)$ so that $i = j$ pass through $r$.
Hence, we have:
$\overline{\delta(r)}=\sum_{i=1}^l \left( |V(C_i)| \cdot \sum_{j \in \{1,\ldots,l\}\setminus \{i\}} |V(C_j)| \right)$.
As a result:
\begin{align}
\frac{\delta_{m\bullet}(r)}{\overline{\delta(r)}} =
\frac
{|V(G)| \cdot V_1}
{\sum_{i=1}^l \left( |V(C_i)| \cdot V_i \right)} \nonumber \\ 
\leq 1 + \frac{V_1 \cdot \sum_{i=1}^l V(C_i)}{\sum_{i=1}^l \left( V(C_i) \cdot V_i \right)}. \label{proof:constantmu1} 
\end{align}
Since for all $i \in [1 .. l]$, $V_i=\Theta(|V(G)|)$, 
therefore for any $j \in [2 .. l]$ we have: $V_j=K_j \cdot V_1$,
where $K_j$ is a constant.
Let $K$ be the minimum of such constants.
Inequality~\ref{proof:constantmu1} yields
\begin{align}
\frac{\delta_{m\bullet}(r)}{\overline{\delta(r)}}
\leq 1 + \frac{V_1 \cdot \sum_{i=1}^l V(C_i)}{K \cdot V_1 \sum_{i=1}^l V(C_i) }
= 1+ \frac{1}{K}. 
\end{align}
Therefore, $\frac{\delta_{m\bullet}(r)}{\overline{\delta(r)}}$ is a constant.
\end{proof}

Theorem~\ref{theorem:constantmu} proposes the general conditions under which $\mu(r)$
is a constant.
In several cases (where $r$ is an important vertex of the graph),
Theorem~\ref{theorem:constantmu} holds and as a result, $\mu(r)$ is a constant.
A simple example is where $r$ a {\em balanced vertex separator} of $G$.
Vertex $x$ is a {\em vertex separator} for $G$,
if there exist vertices $u$ and $v$ that belong to distinct components of $G \setminus x$,
or if $G \setminus x$ contains less than two vertices.
Vertex separator $x$ is {\em balanced}, if at least two components in $G \setminus x$ contains
$\Theta(|V(G)|)$ vertices\footnote{Note that in order to cover more cases,
balanced vertex separator defined here is more general than its
tradition definitional in graph theory.
}.

\subsection{A joint-space MCMC sampler}
\label{sec:mcmcsampler}

In this section, 
we present a MCMC sampler to estimate the ratios of betweenness scores of
the vertices in a set $R \subset V(G)$.
Each state of this sampler
is a pair $(\mathfrak r,\mathfrak v)$,
where $\mathfrak r\in R$ and $\mathfrak v \in V(G)$.
Since this sampler is defined on the joint space $R$ and $V(G)$,
we refer to it as {\em joint-space MCMC sampler}.
Given a state $\mathfrak s$ of the chain,
we denote by $\mathfrak{s.r}$ the first element of $\mathfrak s$,
which is a vertex in $R$;
and by $\mathfrak{s.v}$ the second element of $\mathfrak s$,
which is a vertex in $V(G)$.

Our joint-space MCMC sampler consists of the following steps:
\begin{itemize}
\item 
First, we choose a pair $\left\langle \mathfrak r_0,\mathfrak v_0\right\rangle$,
as the initial state,
where $\mathfrak r_0$ and $\mathfrak v_0$ are chosen uniformly at random from 
$R$ and $V(G)$, respectively. 

\item
Then, at each iteration $t$, $1 \leq t \leq T$:
\begin{itemize}
\item
Let $\mathfrak s(t)$ be the current state of the chain.
\item
We choose elements $\mathfrak r(t) \in R$ and $\mathfrak v(t) \in V(G)$,
uniformly at random.
\item
With probability
\begin{equation}
\min \left\{ 1,\frac{\delta_{\mathfrak v(t)\bullet}(\mathfrak r(t))}{\delta_{\mathfrak s(t).\mathfrak{v} \bullet}(\mathfrak s(t).\mathfrak{r})} \right\} 
\end{equation}
we move from state $\mathfrak s(t)$ to the state
$\left\langle \mathfrak r(t),\mathfrak v(t)\right\rangle$.

\end{itemize}
\end{itemize}

Techniques similar to our joint-space MCMC sampler have been used in statistical physics
to estimate free energy differences \cite{bennett76}.
Our joint-space MCMC sampler is a Metropolis-Hastings algorithm that
possess a unique stationary distribution \cite{Meyn1993,Gilks96}
defined as follows:
\begin{equation}
\label{eq:completeoptimalprobability}
\mathbb P \left[r,v\right]=
\frac{\delta_{v\bullet}(r)}{\sum_{r'\in R} \sum_{v'\in V(G)}\delta_{v'\bullet}(r')}.
\end{equation}
All samples that have a specific value $r$ for their $\mathfrak{r}$ component form
an Independence Metropolis-Hastings chain that possess the stationary distribution
defined in Equation~\ref{eq:optimalsampling}.
Samples drawn by our MCMC and joint-space MCMC samplers are non-iid.
In Theorem~\ref{prop:expectedvalue},
we show how our joint-space MCMC sampler can be used to estimate
the ratios of betweenness scores of the vertices in $R$.


%
%

\begin{theorem}
\label{prop:expectedvalue}
In our joint-space MCMC sampler,
for any two vertices $r_i,r_j \in R$,
we have:
\begin{align}
\label{eq:proof1_4}
\frac{BC(r_i)}{BC(r_j)} =
\frac{\mathbb E_{\mathbb P_{r_j}\left[v\right]}
\left[\min \left\{ 1,\frac{\delta_{v\bullet}(r_i)}{\delta_{v\bullet}(r_j)}\right\} \right] } 
{\mathbb E_{\mathbb P_{r_i}\left[v\right]}
\left[\min\left\{ 1,\frac{\delta_{v\bullet}(r_j)}{\delta_{v\bullet}(r_i)}\right\} \right] }
\end{align}
where $\mathbb E_{\mathbb P_{r_i} \left[v\right]}$
(respectively $\mathbb E_{\mathbb P_{r_j} \left[v\right]}$)
denotes the expected value with respect to $\mathbb P_{r_i}\left[v\right]$
(respectively $\mathbb P_{r_j}\left[v\right]$).
\end{theorem}
\begin{proof}

In our joint-space MCMC sampler,
consider all the transitions
that keep the $\mathfrak v$ element, but change the $\mathfrak r$ element.
According to the \textit{detailed balance property} of 
Metropolis-Hastings algorithms,
for any $v \in V(G)$, we have
\begin{align}
\label{eq:proof1_2}
\mathbb P\left[r_i,v\right] \times TR(\left\langle r_j,v\right\rangle \mid \left\langle r_i,v \right\rangle) &= 
\mathbb P\left[r_j,v\right] \times TR(\left\langle r_i,v\right\rangle \mid \left\langle r_j,v \right\rangle)
\end{align}
Using Equation~\ref{eq:completeoptimalprobability}, we get
\begin{align}
\label{eq:proof1_2}
  \delta_{v\bullet}(r_i) \times
  \min\left\{ 1,\frac{\delta_{v\bullet}(r_j)}{\delta_{v\bullet}(r_i)} \right\} &=  
 \delta_{v\bullet}(r_j) \times
  \min \left\{ 1,\frac{\delta_{v\bullet}(r_i)}{\delta_{v\bullet}(r_j)} \right\}
\end{align}

Equation~\ref{eq:proof1_2} holds for all $v \in V(G)$.
If we sum both sides of Equation~\ref{eq:proof1_2} with respect to $v \in V(G)$
and then, divide/multiply by $BC(r_i)$ and $BC(r_j)$, we get
\begin{align}
\label{eq:proof1_3}
BC(r_i) \sum_{v\in V(G)} 
\left( \underbrace{\frac{\delta_{v\bullet}(r_i)}{BC(r_i)}}_{\mathbb P_{r_i}\left[ v\right]} \times
  \min \left\{ 1,\frac{\delta_{v \bullet}(r_j)}{\delta_{v\bullet}(r_i)} \right\} \right)  &= \nonumber  
BC(r_j) \sum_{v \in V(G)} 
\left( \underbrace{\frac{\delta_{v\bullet}(r_j)}{BC(r_j)}}_{\mathbb P_{r_j}\left[v\right]} \times
  \min \left\{ 1,\frac{\delta_{v\bullet}(r_i)}{\delta_{v\bullet}(r_j)} \right\} \right)
\end{align}
which yields Equation~\ref{eq:proof1_4}.  
\end{proof}

Let $r_i,r_j \in R$, and $M(i)$ and $M(j)$ be the multi-sets
of samples taken by our joint-space MCMC sampler whose $\mathfrak r$ components are respectively $r_i$ and $r_j$.
Equation~\ref{eq:proof1_4} suggests to estimate $\frac{BC(r_i)}{BC(r_j)}$ as the ratio:
\begin{equation}
\label{eq:unbiasedrelativebetweenness}
\frac
{\frac{1}{|M(j)|} \times \sum_{\mathfrak s \in M(j)} \min \left\{1,\frac{\delta_{\mathfrak s.\mathfrak v}(r_i)}{\delta_{\mathfrak s.\mathfrak v}(r_j)}\right\} }
{\frac{1}{|M(i)|} \times \sum_{\mathfrak s \in M(i)} \min \left\{1,\frac{\delta_{\mathfrak s.\mathfrak v}(r_j)}{\delta_{\mathfrak s.\mathfrak v}(r_i)}\right\} }.
\end{equation}

We use Equation~\ref{eq:unbiasedrelativebetweenness}
to estimate the \textit{ratio of the betweenness scores} of $r_i$ and $r_j$.
We then define the \textit{relative betweenness score} of $r_i$ respect to $r_j$,
denoted by $BC_{r_j}(r_i)$,
as follows:
\begin{equation}
\label{eq:relativebetweenness}
BC_{r_j}(r_i)=\frac{1}{|V(G)|}
\sum_{v \in V(G)} \min \left\{ 1,\frac{\delta_{v\bullet}(r_i)}{\delta_{v\bullet}(r_j)} \right\}.
\end{equation}
When we want to compare betweenness centrality of vertices $r_i$ and $r_j$,
using {\em relative betweenness score} makes more sense than using
the {\em ratio of betweenness scores}.
In relative betweenness centrality, for each $v \in V(G)$,
the ratio of the dependency scores of $v$ on $r_i$ and $r_j$ is computed and in the end,
all the ratios are summed.
Hence, for each vertex $v$ independent from the others,
the effects of $r_i$ and $r_j$ on the shortest paths starting from $v$
are examined\footnote{Note that
the notion of {\em relative betweenness score} can be further extended and presented as follows:
$BC_{r_j}(r_i)=\frac{1}{|V(G)|\cdot\left(|V(G)|-1\right)}
\sum_{v \in V(G)}\sum_{t \in V(G)\setminus \{v\}} \min \left\{ 1,\frac{\delta_{vt}(r_i)}{\delta_{vt}(r_j)} \right\}$.}.
In the following, we show that 
the numerator of Equation~\ref{eq:unbiasedrelativebetweenness},
i.e.,
\[\frac{1}{|M(j)|} \sum_{\mathfrak s \in M(j)} \min \left\{1,\frac{\delta_{\mathfrak s.\mathfrak v}(r_i)}{\delta_{\mathfrak s.\mathfrak v}(r_j)}\right\},\]
can accurately estimate $BC_{r_j}(r_i)$.
We refer to the estimated value of $BC_{r_j}(r_i)$ as
$\ddddot{BC}_{r_j}(r_i)$.
Our proof is similar to the analysis we gave in Section~\ref{sec:mcmcsampler0}
and is based on the theorems presented
in \cite{Latuszynski_Miasojedow_Niemiro_2012} and \cite{mh}.
For a pair of vertices $r_i,r_j \in R$, 
in Theorem~\ref{theorem:ourerrorbound2},
we use Inequalities~\ref{eq:uniformlyergodic}, \ref{eq:expinequalty}
and \ref{eq:independenceuniformlyergodic}
to derive an error bound for $\ddddot{BC}_{r_j}(r_i)$.

\begin{theorem}
\label{theorem:ourerrorbound2}
Let $r_i,r_j \in R$,
$M(j)$ be the multi-set of samples whose $\mathfrak r$ components are $r_j$,
and $\overline{\delta(r_j)}$ be
the average of dependency scores of vertices 
in $V(G)$ on $r_j$, i.e.,
\[ \overline{\delta(r_j)}= \frac{\sum_{v\in V(G)}\delta_{v\bullet}(r_j)}{|(V(G)|}.\]
Suppose that 
there exists some value $\mu(r_j)$ such that 
for each vertex $v \in V(G)$, the following holds:
\begin{equation}
\label{eq:averagebcbound2}
\delta_{v\bullet}(r_j) \leq \mu(r_j) \times \overline{\delta(r_j)}.
\end{equation}
Then, 
for a given $\epsilon \in \mathbb R^+$,
by our joint-space MCMC sampler
and starting from any arbitrary initial state, we have
\begin{align}
\label{eq:errorbound2}
\mathbb P\left[ |\ddddot{BC}_{r_j}(r_i)-BC_{r_j}(r_i)| > \epsilon \right] 
\leq 2 \exp{ \left\{ - \frac{|M(j)|-1}{2} 
\left( \frac{2\epsilon}{\mu(r_i)} -
\frac{3}{|M(j)|-1} \right)^2  \right\} }.
\end{align}
\end{theorem}

\begin{proof}
As discussed in \cite{bennett76},
our joint-space MCMC sampler can be implemented as $|R|$ parallel
MCMC samplers, each defined on the state space $V(G)$.
In this way, samples whose $\mathfrak r$ components are $r_j$
(i.e., samples that are in $M(j)$) are drawn from an 
Independence Metropolis Hastings MCMC whose proposal density is $\frac{1}{|V(G)|}$.
Inequation \ref{eq:averagebcbound} yields that
\[\frac{\delta_{v\bullet}(r_j)}{\overline{\delta(r_j)} } =
\frac{|V(G)| \cdot \delta_{v\bullet}(r_j)}{\sum_{v\in V(G)}\delta_{v\bullet}(r_j)}
\leq \mu(r_j),\]
which yields 
\begin{equation}
\frac{1}{|V(G)|}  \geq \frac{1}{\mu(r_j)} \times
\frac{\delta_{v\bullet}(r_j)}{\sum_{v\in V(G)}\delta_{v\bullet}(r_j)}.
\end{equation}

Hence, the MCMC formed by the samples in $M(j)$
satisfies Inequation~\ref{eq:independenceuniformlyergodic},
where the proposal density is
$\frac{1}{|V(G)|}$,
$\mathbb P \left[ \cdot \right]$ is $\mathbb P_{r_j}\left[v\right]$
defined in Equation~\ref{eq:optimalsampling}
and $\beta$ is $\frac{1}{\mu(r_j)}$.
Now, if we apply Inequation~\ref{eq:expinequalty}
on the MCMC formed by the samples in $M(j)$
with:
$n = |M(j)|$, 
$S=V(G)$, and 
$f(v)=
\min \left\{1, \frac{\delta_{v\bullet}(r_i)}{\delta_{v\bullet}(r_j)} \right\}$,
we get:
$\theta = \frac{1}{|V(G)|} \sum_{v\in V(G)} f(v) = BC_{r_j}(r_i)$,
$\hat{\theta}_n = \frac{1}{|M(j)|} \sum_{\mathfrak s \in M(j)} f(\mathfrak s.\mathfrak v) =  \ddddot{BC}_{r_j}(r_i)$,
$\parallel f \parallel_{sp} = 1$, and $\lambda = \frac{1}{\mu(r_j)}$.
Putting these values into Inequation~\ref{eq:expinequalty}, we get Inequation~\ref{eq:errorbound2}. 
\end{proof}

Similar to Inequation~\ref{eq:errorbound},
Inequation~\ref{eq:errorbound2} does not depend on the initial state and
it holds without need for burn-in. 
Furthermore,
for given values $\epsilon \in \mathbb R^+$ and $\delta \in (0,1)$, if we have
\begin{equation}
\label{eq:epsilondeltaT}
|M(j)|  \geq  \frac{{\mu(r_j)}^2}{2\epsilon^2}   \ln{\frac{2}{\delta}} 
\end{equation}
then, our joint-space MCMC sampler can estimate relative betweenness score of $r_i$ respect to $r_j$  
within an additive error $\epsilon$ with probability $\delta$.
Finally, similar to our discussion in Section~\ref{sec:mcmcsampler0},
in several cases
$\mu(r_j)$ is a constant and as a result, our joint-space MCMC sampler
can give an $(\epsilon,\delta)$-approximation of
relative betweenness score of $r_i$ respect to $r_j$
using only a constant number of samples.

\section{Conclusion and future work}
\label{sec:conclusion}

In this paper, first
given a network $G$ and a vertex $r \in V(G)$,
we proposed a Metropolis-Hastings MCMC algorithm that samples from the
space $V(G)$ and estimates betweenness score of $r$.
We showed that our MCMC sampler provides an 
$(\epsilon,\delta)$-approximation, where the number of
required samples depends on the position of $r$ in $G$
and in several cases, it is a constant.
Then,
given a network $G$ and a set $R \subset V(G)$,
we presented a Metropolis-Hastings MCMC sampler that samples from the
joint space $R$ and $V(G)$ and estimates relative betweenness scores of the vertices in $R$.
We showed that for any pair $r_i, r_j \in R$,
the ratio of the expected values of the estimated relative betweenness scores
of $r_i$ and $r_j$ respect to each other is equal to
the ratio of their betweenness scores.
We also showed that our joint-space MCMC sampler provides an 
$(\epsilon,\delta)$-approximation
of the relative betweenness score of $r_i$ respect to $r_j$,
where the number of required samples depends on the position of $r_j$ in $G$
and in many cases, it is a constant.

The sampling techniques presented in this paper have similarity to
the techniques used in statistical physics to estimate e.g., energy differences.
Our current paper takes the first step in bridging these two domains.
This step can be further extended by proposing algorithms similar to our work that estimate other network indices.
As a result, a novel family of techniques might be introduced into the field of network analysis,
that are used to estimate different indices.

\bibliography{allpapers}   

\end{document}